%% file: main.tex
\documentclass[submission, copyright, creativecommons, sharealike, noncommercial]{eptcs}

\usepackage{graphicx}
\usepackage{amsthm,amsmath,amssymb}
\usepackage{mathtools}
\usepackage{proof}
\usepackage{eptcsstyle/breakurl}

\usepackage{todonotes}
\usepackage{stmaryrd}
\usepackage{mathpartir}
\usepackage[capitalise]{cleveref}
\usepackage{microtype}
\usepackage{float}

\usepackage{tikz}
\usepackage{tikz-cd}
\usetikzlibrary{decorations.pathreplacing}
\usepackage{tikz-qtree}
\usepackage{tikzit}
\usepackage{circuitikz}
\input{styles.tikzstyles}

\usepackage{hyperref}
\hypersetup{
    colorlinks=true,
    linkcolor=blue,
    filecolor=magenta,
}
\usepackage{macros}

\title{
  Translating Extensive Form Games\\to Open Games with Agency
}
\author{
    Matteo Capucci,
    Neil Ghani,\\
    J{\'e}r{\'e}my Ledent,
    Fredrik Nordvall Forsberg
    \institute{Mathematically Structured Programming Group, University of Strathclyde, \'Ecosse Libre}
}
\def\titlerunning{
    Translating Extensive Form Games to Open Games With Agency
}
\def\authorrunning{
    M. Capucci, N. Ghani, J. Ledent, F. Nordvall Forsberg
}

\providecommand{\event}{ACT 2021}

\begin{document}

    \maketitle

    \begin{abstract}
        We show open games cover extensive form games with both perfect and imperfect information. Doing so forces us to address two current weaknesses in open games: the lack of a notion of player and their agency within open games, and the lack of choice operators. Using the former we construct the latter, and these choice operators subsume previously proposed operators for open games, thereby making progress towards a core, canonical and ergonomic calculus of game operators. Collectively these innovations increase the level of compositionality of open games, and demonstrate their expressiveness.
    \end{abstract}

    \input{src/introduction.tex}
    \input{src/extensive-form-games.tex}
    \input{src/open-games-with-players.tex}
    \input{src/operators-on-open-games.tex}
	\input{src/translation.tex}
    \input{src/conclusion.tex}

    \bibliography{bibliography.bib}

    % \newpage
    % \appendix
    % \input{src/appendix.tex}
\end{document}

%% file: styles.tikzstyles
\tikzstyle{medium box}=[fill=white, draw=black, shape=rectangle, minimum width=1.75cm, minimum height=1cm]
\tikzstyle{dot}=[fill=black, draw=black, shape=circle, inner sep=0pt, minimum size=3pt]
\tikzstyle{horizontal box}=[fill=white, draw=black, shape=rectangle, minimum width=3cm, minimum height=1cm]
\tikzstyle{small box}=[fill=white, draw=black, shape=rectangle, minimum width=0.5cm, minimum height=0.5cm]
\tikzstyle{smallish box}=[fill=white, draw=black, shape=rectangle, minimum width=0.8cm, minimum height=0.3cm]
\tikzstyle{large box}=[fill=white, draw=black, shape=rectangle, minimum width=2.5cm, minimum height=1.5cm]
\tikzstyle{long horizontal box}=[fill=white, draw=black, shape=rectangle, minimum width=6cm, minimum height=0.75cm]
\tikzstyle{left-discard}=[fill=white, draw=black, shape=circle, ground, rotate=-90]
\tikzstyle{square}=[fill=white, draw=black, shape=rectangle, minimum width=1cm, minimum height=1cm]
\tikzstyle{up-discard}=[fill=white, draw=black, shape=circle, rotate=180, ground]
\tikzstyle{medium horizontal box}=[fill=white, draw=black, shape=rectangle, minimum width=4.5cm, minimum height=1cm]
\tikzstyle{long dotted box}=[fill=white, draw=black, shape=rectangle, minimum width=7cm, minimum height=0.75cm, dotted, fill opacity=0, text opacity=1, semithick]
\tikzstyle{large dotted box}=[fill=white, draw=black, shape=rectangle, minimum width=2.5cm, minimum height=1.5cm, dotted, fill opacity=0, text opacity=1, semithick]

\tikzstyle{dashed line}=[-, dashed]
\tikzstyle{arrow}=[->, font={\scriptsize}, inner sep=2pt]
\tikzstyle{curly braces}=[-, decorate, decoration={{brace,amplitude=3pt,raise=4pt}}, font={\scriptsize}]

%% file: src/introduction.tex
\section{Introduction}

Game theory is the study of how players make decisions.
Crucially, there is typically more than one player and players have competing and even divergent goals: Nash equilibria describe a choice of strategy for each player such that, if any one player were to choose a different strategy, the outcome for that player would be worse.
Thus no player is incentivized to change strategy and an equilibrium is established.

There are many ways to formalise games; one of the most commonly used are \emph{extensive form games}.
These are decorated trees whose path structure reflects decisions taken sequentially, whose internal nodes reflect choices to be made by players, and whose leaves represent the utility accrued via the choices made.
Extensive form games have a device (\emph{information sets}) which allows users to specify situations in which players may face decisions of which they do not know the exact location in the tree.
Such situations happen, for example, if players make decisions in parallel or if some information is hidden (e.g.~in a card game, the opponent's hand is unknown).
When this happens, we speak of \emph{imperfect information games}, to be contrasted with \emph{perfect information} ones.

This formalism for game theory suffers from the significant drawback of not being \emph{compositional} in that one cannot construct larger games ---including their equilibria--- from smaller ones.
Open games~\cite{julesGames} are a new foundation for game theory which offers a simultaneous treatment of the sequential and concurrent nature of games and, crucially, does so compositionally.
Apart from compositionality, working under the principles of categorical logic and type theory means that rigour is increased, models get conceptually sharper, and software is more easily extracted from open games \cite{julesHaskell}.
Despite this, a convincing procedure to convert extensive form games into open games has been lacking until now.
This paper provides such a translation, and in doing so introduces some important innovations:
\begin{itemize}
  \item Open games have so far lacked an explicit notion of player, which is fundamental in game theory, and this lack hinders attempts to relate open games to extant game theory.
  To deal with this problem, we introduce \emph{open games with agency}, an improved version of open games based on parametrised lenses and selection functions.
  \item We define choice operators for open games with agency.
  Choice operators have so far been missing from the operators which are supported by open games, although limited forms have been used in the software implementation of open games.
  \item We show that extensive form games can be translated into open games with agency preserving the decision structure and with the same equilibria.
  Our proof is a significant improvement on the 8-page proof in~\cite{julesPhD} for perfect information games, while the result for imperfect information games is new.
  Our translations rely crucially on the new choice operators and reparametrisation operations.
\end{itemize}

\paragraph{Outline of the paper.}
In~\cref{sec:ext_games}, we introduce extensive form games as an inductive data type. In~\cref{sec:open_games}, we introduce the new open games with agency, and~\cref{sec:operators} discusses the associated operator calculus.~\cref{sec:ext_translation} contains the translation from the former to the latter.
Finally, in~\cref{sec:conc}, we draw conclusions and discuss further work.

\paragraph{Acknowledgements.}
The authors would like to thank Scott Cunningham, Clemens Kupke, Glynn Winskel and Radu Mardare for their thoughtful feedback, help and discussions, as well as Jules Hedges, Bruno Gavranovi\'{c} and the rest of MSP group for invaluable input and conversations.

\paragraph{Notation.} We use ${\cal P} : \msf{Set} \rightarrow \msf{Set}$ for the covariant powerset functor mapping a set to its set of subsets.
Given a set-indexed collection of sets $Y : I \rightarrow \msf{Set}$, the dependent sum $(\Sigma i:I)\,Y(i)$ is the disjoint union of all of the sets in the collection, while the dependent function space $(\Pi i:I)\,Y(i)$ is the set of functions mapping an input $i \in I$ to an element of $Y(i)$.
We may also use ``Agda~\cite{norell:thesis} notation'' $(i : I) \to Y(i)$ for the dependent function space.
Note that a function $k : (\Sigma i:I)\,Y(i) \to R$ is the same as an $I$-indexed collection of functions $k_i : Y(i) \to  R$.
Morphisms compose diagrammatically.
We write $\N$ for the set of natural numbers, $\N^+$ for the positive natural numbers, and $[n] = \{0, \ldots, n-1\}$ for the canonical $n$-element set.

%%% Local Variables:
%%% mode: latex
%%% TeX-master: "../main"
%%% End:

%% file: src/extensive-form-games.tex
%\vspace{-3ex}
\section{Extensive Form Games}
\label{sec:ext_games}

\paragraph{Extensive form games with perfect information.} Perfect information games are defined in~\cite{eogt} using a non-inductive formulation. We present this definition categorically as it is more compact, supports recursive function definitions and inductive forms of reasoning, and smoothly generalises to richer semantic spaces, e.g.\ metric, probabalistic and topological spaces. Throughout this section, let $P$ be a set of players and $R$ a preordered set of rewards. We write $R^P = P \to R$ for the set of utility functions.

\begin{definition}
    The set $\PETree$ of extensive form game trees with perfect information is the initial algebra of the functor $F:\msf{Set} \to \msf{Set}$ defined by
    \begin{equation*}
      FX = R^P + P \times (\Sigma\, n:\N^+)([n] \to X).
    \end{equation*}
    This supports the Haskell-like data type
    \begin{equation*}
      \mbox{data} \;\; \PETree = \Leaf \;\; R^P\quad|\quad \Node \;\; P \;\; (n : \N^+)\;\; ([n] \to  \PETree)
    \end{equation*}
\end{definition}

Thus a tree $T \in \PETree$ is either a leaf containing a reward vector, or an internal node labelled with a player $p \in P$ (who is to play at that point in the game) and of \emph{arity} $n \in \N^+$ (the number of different moves available to player $p$ at that point). As mentioned, being an inductive type $\PETree$ supports recursive function definitions by pattern matching. We can take advantage of this to write down classical definitions somewhat more nimbly.
For instance, the set of strategies of player $p \in P$ is defined:
\begin{align*}
  & \strategies_\PET : \PETree \to P \to \msf{Set} \\
  & \strategies_\PET \; (\Leaf \; v)\; p = \One \\
  & \strategies_\PET \; (\Node \; q \; n \; f)\; p = (\msf{if}\ p \equiv q\ \msf{then}\ [n]\ \msf{else}\ \One) \times (\Pi\,{m \in [n]})\, (\strategies_\PET\; (f\; m)\; p)
\end{align*}
For a given tree $T \in \PETree$ and player~$p$, a strategy $\omega \in \strategies_\PET\;T\;p$ consists of a choice, for every $n$-ary node labeled by~$p$, of a branch $m \in [n]$.
Note that, in accordance with the game theory literature, player~$p$ chooses a move in all their nodes in the tree --- even in those that are not reachable according to their previous choices. This is because players need to consider alternatives to chosen strategies both for themselves and others to counterfactually evaluate equilibria.

A \emph{strategy profile} consists of a strategy for each player $p \in P$. For conciseness, we write $\Omega_p = \strategies_\PET\;T\;p$ for the set of strategies of player~$p$, and $\Omega = \Pi_p \Omega_p$ for the set of strategy profiles of the game~$T$.
It is not hard to see that the set $\Omega$ can equivalently be defined inductively  as follows:
\begin{align*}
  & \profiles_\PET : \PETree \to \msf{Set}\\
  & \profiles_\PET \; (\Leaf \; v) = \One\\
  & \profiles_\PET \; (\Node \; p \; n \; f) = [n] \times (\Pi\,{m \in [n]}) \, (\profiles_\PET\; (f\; m))
  \intertext{We also define the set $\paths_\PET\; T$ of possible game trajectories, which eventually lead to the reward vector stored at the leaf (as computed by $\payoff_\PET$ below).}
  & \paths_\PET : \PETree \to \msf{Set} \\
  & \paths_\PET \; (\Leaf \; v) = \One \\
  & \paths_\PET \; (\Node \; p \; n \; f) = (\Sigma\,{m \in [n]}) \; (\paths_\PET\; (f\; m))
\end{align*}
A game trajectory $\pi$ determines a payoff $\payoff_\PET \; T \;\pi$ for each player, and a strategy profile $\omega$ determines a unique path $\play_\PET \; T \; \omega$ from the root to a leaf, obtained by playing the strategies against each other.
\begin{align*}
  & \payoff_\PET : (T : \PETree) \to (\paths_\PET\;T) \to R^P\\
  & \payoff_\PET \; (\Leaf \; v) \; \unit = v\\
  & \payoff_\PET \; (\Node \; p \; n \; f) \; (m, \pi) = \payoff_\PET\; (f\; m) \; \pi\\[2ex]
  & \play_\PET : (T : \PETree) \to (\profiles_\PET\;T) \to (\paths_\PET\; T) \\
  & \play_\PET \; (\Leaf \; v) \; \unit = \unit \\
  & \play_\PET \; (\Node \; p \; n \; f) \; (m, \phi) = (m, \play_\PET\; (f\; m)\; (\phi\; m))
\end{align*}

\noindent Fix a game $T \in \PETree$.
Given a strategy profile $\omega \in \Omega$ and a $p$-strategy $\omega_p \in \Omega_p$, we write $\omega[p \leftarrow \omega_p]$ for the strategy profile obtained by replacing the $p$-th component of~$\omega$ by $\omega_p$.
Recall that the set~$R$ of rewards is ordered.

\begin{definition}[Nash equilibrium]
  The \emph{Nash equilibrium} predicate $\msf{Nash} : (T : \PETree) \to {\cal P} (\profiles_\PET\;T)$ holds for a strategy profile $\omega$ of a game $T$ if and only if for every player $p \in P$, and for every $\omega_p \in \Omega_p$,
  \[
    \payoff_\PET \; T \; (\play_\PET\;T\;\omega) \; p \;\; \geq \;\; \payoff_\PET \; T \; (\play_\PET\;T\;\omega[p \leftarrow \omega_p]) \; p \enspace .
  \]
\end{definition}

\noindent In plain words, a strategy profile is a Nash equilibrium if the payoff each player receives cannot be improved by that player unilaterally changing their strategy.

\begin{example}
  Two perfect information games are shown below.
  Nodes are labelled by players $p_1$, $p_2$ or~$p_3$.
  All nodes must choose between two moves, $L$ and~$R$. Utility at the leaves are for $p_1, p_2, p_3$ in that order.
  \begin{center}
  \begin{tikzpicture}[every level 0 node/.style={circle,draw,inner sep=1.2,fill=black},
  every level 1 node/.style={circle,draw,inner sep=1.2,fill=black},
  level distance=.5in, sibling distance=1em,
  edge from parent path={(\tikzparentnode) -- (\tikzchildnode)}]
  \Tree [.\node[label=above:{$p_1$}] {};
    \edge node [pos=0.6, auto=right] {L}; [.\node[label=left:{$p_2$}] {};
      \edge node [pos=0.8, auto=right] {L}; [.\node {$(1,3,1)$};]
      \edge node [pos=0.8, auto=left] {R}; [.\node {$(4,0,4)$};]
    ]
    \edge node [pos=0.6, auto=left] {R}; [.\node[label=right:{$p_3$}] {};
      \edge node [pos=0.8, auto=right] {L}; [.\node {$(0,0,0)$};]
      \edge node [pos=0.8, auto=left] {R}; [.\node {$(8,5,8)$};]
    ]
  ]
  \end{tikzpicture}
  \qquad \qquad
  \begin{tikzpicture}[every level 0 node/.style={circle,draw,inner sep=1.2,fill=black},
  every level 1 node/.style={circle,draw,inner sep=1.2,fill=black},
  level distance=.5in, sibling distance=1em,
  edge from parent path={(\tikzparentnode) -- (\tikzchildnode)}]
  \Tree [.\node[label=above:{$p_1$}] {};
    \edge node [pos=0.6, auto=right] {L}; [.\node[label=left:{$p_2$}] {};
      \edge node [pos=0.8, auto=right] {L}; [.\node {$(1,3)$};]
      \edge node [pos=0.8, auto=left] {R}; [.\node {$(4,0)$};]
    ]
    \edge node [pos=0.6, auto=left] {R}; [.\node[label=right:{$p_1$}] {};
      \edge node [pos=0.8, auto=right] {L}; [.\node {$(0,0)$};]
    \edge node [pos=0.8, auto=left] {R}; [.\node {$(8,5)$};]
    ]
  ]
  \end{tikzpicture}
  \end{center}
  The game on the left has three players, each of them making one decision.
  The strategy profile $(L,L,L)$ is a Nash equilibrium of this game, which yields the utility $(1,3,1)$.
  The game on the right has only two players, with~$p_1$ making two decisions.
  In this second game,  $((L,L),L)$ is not a Nash equilibrium because~$p_1$ can change strategy to $((R,R),L)$ and get a better reward.
  In the first game, even though~$p_1$ and~$p_3$ always get the same reward, they are different players and so cannot similarly coordinate changes to their strategies.
\end{example}

\paragraph{Extensive form Games with Imperfect Information.}
Imperfect information games are a refinement of perfect information games via the addition of \emph{information sets}.
The intuition is that a player only knows that they are at \emph{some} node in the set, but not which one, and not the full state of the game at that point.

As before, we replace the traditional definition of information sets~\cite{eogt} with an inductive data type where information sets are labels on nodes. The conditions that every node in an information set belongs to the same player and that nodes labelled by the same information set should have the same arity now hold by construction.
We represent information sets with the following data:
\begin{itemize}
  \item A set $I$ of \emph{labels} for information sets. Each label is thought of as a set of nodes.
  \item A function $\belongs : I \to P$ that assigns a player to each information set. This ensures that every node labelled with that information set has the same player.
  \item A natural number for each information set $n : I \to \N^+$. This ensures that every node labelled with that information set has the same arity.
\end{itemize}

\begin{definition}
    The set $\IETree$ of \emph{extensive-form game trees with imperfect information} with players in $P$, rewards in $R$, and information sets $(I, \belongs, n)$ is defined as the inductive data type
 \[
   \mbox{data} \;\; \IETree = \Leaf \;\; R^P \quad|\quad \Node \;\; (i : I) \;\; ([n\; i] \to  \IETree)
 \]
\end{definition}
\noindent Notice that every $\PETree$ is an $\IETree$ where we take as information sets the nodes of the $\PETree$ --- that is, no two nodes are in the same information set.
This means that, contrary to the usual definition, an $\IETree$ might have information sets that are never used or even players that never play, though this is clearly pathological and always avoidable.

A strategy for player~$p$ in a tree $T$ is a choice of a move $m \in [n\; i]$ for each information set that occurs in the tree (we write $I \cap T$ for these information sets) and belongs to~$p$.
Note that this means that they must choose the same move for every node in a given information set, by construction.
A strategy profile is a choice of a move for each information set that occurs in the tree:
\[
  \strategies_\IET\;T\;p \coloneqq \prod_{\substack{i \in I \cap T \\ \msf{\scriptsize belongs}\;i = p}} [n\; i]
\qquad \qquad
\profiles_\IET\;T \coloneqq \Pi_{p \in P} (\strategies_\IET\;T\; p) = \Pi_{i \in I \cap T} [n\; i]
\]
All the other notions (paths, payoffs, Nash and SPE equilibria) are defined in the same way as in the case of perfect information games.

\begin{example}
\label{ex:IET}
  Two imperfect information  games are depicted below.
\begin{center}
  \begin{tikzpicture}[every level 0 node/.style={circle,draw,inner sep=1.2,fill=black},
  every level 1 node/.style={circle,draw,inner sep=1.2,fill=black},
  level distance=.5in, sibling distance=1em,
  edge from parent path={(\tikzparentnode) -- (\tikzchildnode)}]
  \Tree [.\node[label=above:{$p_1$}] {};
    \edge node [pos=0.6, auto=right] {$L$}; [.\node[label=left:{$p_2$}] {};
      \edge node [pos=0.8, auto=right] {$\ell$}; [.\node {$(1,4)$};]
      \edge node [pos=0.8, auto=left] {$r$}; [.\node {$(0,0)$};]
    ]
    \edge node [pos=0.6, auto=left] {$R$}; [.\node[label=right:{$p_2$}] {};
      \edge node [pos=0.8, auto=right] {$\ell$}; [.\node {$(5,2)$};]
      \edge node [pos=0.8, auto=left] {$r$}; [.\node {$(0,2)$};]
    ]
  ]
  \end{tikzpicture}
  \qquad \qquad
  \begin{tikzpicture}[every level 0 node/.style={circle,draw,inner sep=1.2,fill=black},
  every level 1 node/.style={circle,draw,inner sep=1.2,fill=black},
  level distance=.5in, sibling distance=1em,
  edge from parent path={(\tikzparentnode) -- (\tikzchildnode)}]
  \Tree [.\node[label=above:{$p_1$}] {};
    \edge node [pos=0.6, auto=right] {$L$}; [.\node[label=left:{$p_2$}] (P2left) {};
      \edge node [pos=0.8, auto=right] {$\ell$}; [.\node {$(1,4)$};]
      \edge node [pos=0.8, auto=left] {$r$}; [.\node {$(0,0)$};]
    ]
    \edge node [pos=0.6, auto=left] {$R$}; [.\node[label=right:{$p_2$}] (P2right) {};
      \edge node [pos=0.8, auto=right] {$\ell$}; [.\node {$(5,2)$};]
      \edge node [pos=0.8, auto=left] {$r$}; [.\node {$(0,2)$};]
    ]
  ]
  \draw[dashed] (P2left) --  (P2right);
  \end{tikzpicture}
\end{center}
The game on the left happens to be a perfect information game,
where player~$p_2$ is able to observe the move made by~$p_1$, and thus has to make two decisions (one for each subtree).
One possible Nash equilibrium of this game is $(L, (\ell,r))$, yielding the utility of $(1,4)$. The game on the right has a non-trivial information set represented by a dashed line: the two nodes labelled by~$p_2$ are in the same information set, so --- crucially --- $p_2$ must play the same strategy irrespective of which move
$p_1$ makes. In effect, they play concurrently. The only Nash equilibrium of this game is $(R,\ell)$, yielding the utility~$(5,2)$.
\end{example}

%%% Local Variables:
%%% mode: latex
%%% TeX-master: "../main"
%%% End:

%% file: src/open-games-with-players.tex
\section{Open Games with agency}
\label{sec:open_games}
\vspace{1ex}
\paragraph{Lenses and parametrised lenses.}
\label{sec:para_lens}
Lenses are bidirectional maps between pairs of sets~\cite{lenses}:
given four sets $X,Y,R,S$, a \emph{lens} L from $(X,S)$ to $(Y,R)$, written $L : \Lens{X}{S}{Y}{R}$, consists of a pair of functions $\get_L : X \to Y$ and $\put_L : X \times R \to S$.
Lenses are morphisms in a category whose objects are pairs of sets.
The tensor product or \emph{parallel composition} of two lenses $L : \Lens{X}{S}{Y}{R}$ and $L' : \Lens{X'}{S'}{Y'}{R'}$ is a lens $L \otimes L' : \Lens{X \times X'}{S \times S'}{Y \times Y'}{R \times R'}$ making $\otimes$ a bifunctor on the category of lenses --- see~\cite{riley2018categories} for details and generalizations.

A \emph{parametrised lens} from $(X,S)$ to $(Y,R)$ with parameters $(P,Q)$ is a lens from $(P \times X, Q \times S)$ to $(Y,R)$. We denote by $\PLens{(P,Q)}{X}{S}{Y}{R}$ such parametrised lenses. Thus, $\PLens{(P,Q)}{X}{S}{Y}{R} = \Lens{P \times X}{Q \times S}{Y}{R}$. Note in particular that an ordinary lens can be considered a parametrised lens with parameters $(\I, \I)$. Crucially, parametrised lenses do not compose in the same way as regular lenses:
given $K : \PLens{(P,Q)}{X}{S}{Y}{R}$ and $L : \PLens{(P',Q')}{Y}{R}{Z}{T}$, their sequential composition has type $K \fatsemi L : \PLens{(P \times P', Q \times Q')}{X}{S}{Z}{T}$.
We usually represent parametrised lenses as boxes with the parameters drawn on top, and arrows indicating the flow of information. This allows to read off the type of a composite as in the picture below:

\begin{figure}[H]
	\sctikzfig[1.2]{para-lens-composition}
\end{figure}

\noindent The composite $K \fatsemi L$ is defined by the following composition of lenses:
\begin{center}
	\begin{tikzcd}
		(P\times P'\times X, Q\times Q'\times S) \arrow[rr, "\id_{(P',Q')}\, \otimes\, K"] &&
		(P' \times Y, Q' \times R) \arrow[r, "L"] &
		(Z,T)
	\end{tikzcd}
\end{center}
\noindent Given a lens $w : \Lens{P'}{Q'}{P}{Q}$ and a parametrised lens $L : \PLens{(P,Q)}{X}{S}{Y}{R}$, the \emph{reparametrisation of $L$ along $w$}, written $\reparam{w}{L} : \PLens{(P',Q')}{X}{S}{Y}{R}$, is defined by:
\begin{center}
	\begin{tikzcd}
		(P'\times X, Q'\times S) \arrow[rr, "w\, \otimes\, \id_{(X,S)}"] &&
		(P \times X, Q \times S) \arrow[r, "L"] &
		(Y,R)
	\end{tikzcd}
\end{center}

\noindent Finally, the parallel composition of parametrised lenses is the same as their parallel composition viewed as regular lenses, so we still denote it by~$\otimes$.
In particular, if $L$ and $L'$ have parameters $(P,Q)$ and $(P',Q')$, respectively, then $L \otimes L'$ has parameters $(P\times P',Q \times Q')$.

Parametrised lenses were first introduced by \cite{baf}, although with reparametrisations trivialized by a quotient. Instead, in this work reparametrisations are of utmost centrality.
See \cite{cybernetics} for a general framework behind this intuition.

\paragraph{Selection functions.}
Given sets $X$ and $S$, a multivalued selection function over $(X,S)$ is an element of the set
$\Sel(X,S) = (X \to S) \to \Pow{X}$.
Intuitively, $X$ is a set of \emph{options} and $S$ a set of \emph{rewards}. A map $X \to S$ is then a \emph{valuation} of said options, and a selection function $\varepsilon \in \Sel(X,S)$ represents a way to select options given the information of their value. The canonical selection function is $\argmax$, which returns those inputs on which a function reaches a maximal value.
For a comprehensive survey on selection functions, with applications to game theory and proof theory, see \cite{escardo-oliva}. By noticing that $X \to S \iso \Lens{X}{S}{\I}{\I}$ (so-called \emph{co-states} of $(X,S)$) and $X \iso \Lens{\I}{\I}{X}{S}$ (\emph{states} of $(X,S)$), $\Sel$ can be defined as a functor $\cat{Lens} \to \Cat$ given on objects by
\[
	\Sel(X,S) := \Lens{X}{S}{\I}{\I} \to \Pow{(\Lens{\I}{\I}{X}{S})}
\]
considered as a category with morphisms given by pointwise inclusion, and on morphisms by \emph{pushforward}
\[
	\Sel((X,S) \overset{L}\to (Y,R))(\varepsilon : \Sel(X,S)) := \lambda (k : \Lens{Y}{R}{\I}{\I})\,.\, \{ x \fatsemi L \suchthat x \in \varepsilon(L \fatsemi k) \}.
\]
Most importantly for this paper, $\Sel$ admits a lax monoidal structure
\[
  - \boxtimes - : \Sel(X, S) \times \Sel(Y, R) \longrightarrow {\Sel(X \times Y, S \times R)}
\]
defined by
\[
  (\varepsilon \boxtimes \eta)\; k \coloneqq \{ (x, y) \suchthat x \in \varepsilon\; (\pi_1 k(-,y)), \ y \in \eta\; (\pi_2 k(x,-)) \} \enspace .
\]
We call $\boxtimes$ the \emph{Nash product} of selection functions (called \emph{sum} in \cite{hedges2018backward}), a terminology that will be motivated in \cref{th:normal_form_nash}. More abstractly, one can define a category $\int \Sel$ whose objects are selection functions $(X,S,\varepsilon)$; see \cite{cybernetics} for details. The pushforward construction above makes the forgetful functor $\int \Sel \rightarrow \msf{Lens}$ sending
$(X,S,\varepsilon)$ to $(X,S)$ into an opfibration --- the Nash product makes this a monoidal opfibration.

\paragraph{Open games with agency.}
We can now give the main definition of this section:

\begin{definition}
\label{def:game}
	An \emph{open game with agency} $\G$ is given by:
	\begin{itemize}
		\item A parametrised lens $\A_\G : \PLens{(\Omega, \Comega)}{X}{S}{Y}{R}$ called the \emph{arena} of the game;
		\item A \emph{selection function} $\varepsilon_\G \in \Sel(\Omega, \Comega)$.
	\end{itemize}
	The set of such games is written $\Game{\Omega}{\Comega}{\varepsilon}{X}{S}{Y}{R}$.
	The set~$X$ is called the set of \emph{states}, $Y$ is the set of \emph{moves}, $R$~the set of \emph{utilities} and~$S$ the set of \emph{coutilities}.
	The set $\Omega$ is called the set of \emph{strategy profiles}, and~$\Comega$ is the set of \emph{reward vectors}.
	The ``get'' component $\Omega \times X \to Y$ of the arena $\A$ is called the \emph{play} function; and the ``put'' of the arena $\Omega \times X \times R \to S \times \Comega$ is called the \emph{coplay} function.
\end{definition}

\noindent Arenas are the compositional part in open games with agency.
The reason we make a distinction between arenas and open games is the highly non-compositional nature of selection functions: most operations on arenas do not extend correctly to selection functions on those arenas (as explained in~\cref{sec:recipe}).
However, often the choice of selection functions is somehow canonical and can be done after constructing the arena, which makes compositional constructions on arenas the most important step in the construction of open games with agency.

In open games with agency, players are still kept implicit but they can be introduced explicitly by a suitable indexing of strategies, rewards and selection functions:
\begin{definition}
	Let $P$ be a set. An \emph{arena with players~$P$} is an arena whose set of strategy profiles is of the form $\Omega = \Pi_{p \in P} \Omega_p$, and whose set of reward vectors is $\Comega = \Pi_{p \in P} \Comega_p$.
	An \emph{open game with players~$P$} is given by an arena with players $P$, and a selection function of the form $\varepsilon = \Boxtimes_{p \in P} \varepsilon_p$ for a $P$-indexed family of selection functions $\varepsilon_p : \Sel(\Omega_p, \Comega_p)$.
\end{definition}

A game with players can be pictured as shown below, making clear the separation of strategy profiles and reward vectors into $P$-indexed products of sets $(\Omega_p, \Comega_p)$.

\begin{figure}[H]
	\sctikzfig[1.5]{open-game}
\end{figure}

\paragraph{Equilibria.} A \emph{context} for an open game with agency $\G$ is a pair $(x, k)$ of a state and a costate for $\A_\G$. Given a context, we can define the \emph{equilibria} of $\G$ to be the set
\[
	\Eq_\G \; x \; k \coloneqq \varepsilon_\G(x \fatsemi \A_\G \fatsemi k)
\]
%where we implicitly used the fact that $\Omega \rightarrow \Comega = \PLens{(\Omega,\Comega)}{\I}{\I}{\I}{\I}$.
This also shows how open games with agency can be made into open games as defined in~\cite{julesGames}, which consist of a parametrised lens $\A : \PLens{(\Omega, 1)}{X}{S}{Y}{R}$ and an equilibrium predicate $\Eq:X \times (Y \rightarrow R) \rightarrow \Pow(\Omega)$.

Let us now show that the set $\Eq$ deserves the name of set of \emph{Nash equilibria} when dealing with the simplest (mathematically speaking) form of games, \emph{normal-form games}~\cite[Def.~1.2.1]{eogt}.
An \emph{$n$-player normal-form game} is given by a set $A = A_1 \times \ldots \times A_n$ of \emph{actions}, and $u = (u_1, \ldots, u_n)$ a tuple of utility functions $u_p : A \to \RR$.
The goal of each player is to maximize their utility.
We can represent normal-form games as open games with players --- in \cref{sec:ext_translation}, we will similarly give a translation from extensive form games to open games such that the equilibria coincide with Nash equilibria of the original game.

\begin{theorem}
	\label{th:normal_form_nash}
	Let $(A, u)$ be an $n$-player normal-form game. Let
	\[
		\mathcal{N} : \Game{A}{\RR^P}{\Boxtimes_{p \in P} \argmax^\RR_{A_p}}{\One}{\I}{A}{\RR^P}
	\]
	be the open game with set of players $P \coloneqq \{1, \ldots, n\}$ described as follows: its arena $\A_{\mathcal N}$ is the parametrised lens whose ``get'' and ``put'' parts are both identities, and the selection function is the Nash product of $\argmax_{A_p}^\RR$ (since each player is utility-maximizing).
	Then a strategy profile $a \in A$ is a Nash equilibrium for $(A, u)$ if and only if it belongs to $\Eq_{\mathcal N}(\bullet, u)$.
\end{theorem}
\begin{proof}
  It suffices to analyze what $\Eq_{\mathcal N}(\bullet, u)$ amounts to.
  By definition of $\boxtimes$, we have $a \in \Eq_{\mathcal N}(\bullet, u)$ if and only if, for every $p \in P$,
  $
    a_p \in \argmax^\RR_{a'_p \in A_p} u_p (a[p \leftarrow a'_p])
  $.
  Hence $\Eq_{\mathcal N}(\bullet, u)$ contains exactly the strategies for which no player has an incentive to deviate --- the definition of Nash equilibria.
\end{proof}

%%% Local Variables:
%%% mode: latex
%%% TeX-master: "../main"
%%% End:

%% file: src/operators-on-open-games.tex
\section{Operators on open games with agency}
\label{sec:operators}

Open games with agency is a compositional model of game theory. Hence we want operators for building complex open games from simpler ones.
We divide these operators into operators on arenas and operators on games: the difference is in the way selection functions are accounted for.

\paragraph{Sequential composition, parallel composition, and reparametrisation of arenas.}
Arenas being just parametrised lenses, we already defined sequential and parallel composition for them, along with reparametrisation (\cref{sec:para_lens}).
This latter operation is important as it can be used to introduce dependencies between causally distant parts of an arena --- this is our main motivation in defining arenas as parametrised lenses.
We give some examples of how this can be used here; we will also exploit this idea heavily in \cref{sec:ext_translation}.

\begin{example}
\label{ex:addRewards}
  Assume that some player receives a tuple of rewards $\Comega = \RR^n$ (for instance, representing the money earned in each round of an $n$-round game).
  We may want to reparametrise it along a lens ${\resum : (\Omega, \RR) \to (\Omega, \RR^n)}$ whose ``get'' map is the identity, and whose ``put'' sums the rewards, in order to forbid the selection function (which will eventually be attached to the arena) from seeing each separate component.
  This can also be useful (together with regrouping players; see \cref{ex:identifying-players}) to model a coalition in a game, where two players share the profits.
\end{example}

\begin{example}
\label{ex:diag}
  Consider an arena with strategies $\Omega = M \times M$, i.e., the player has to choose two independent moves $m_1, m_2 \in M$.
  We may want to reparametrise this arena along a lens $\clone : (M, \Comega) \to (M \times M, \Comega)$ whose ``get'' map is the diagonal $\Delta : M \to M \times M$, and ``put'' is $\pi_2$, carrying over the reward, so that a strategy now is a single move $m$ which gets cloned in the original arena.
  This will be useful in order to model extensive form games with imperfect information, where a player has to be constrained to choose the same move in all nodes belonging to the same information set.
\end{example}

\paragraph{Composing and regrouping arenas with players.}
We observe that when composing arenas with players $P$ and $Q$, the resulting arena has players $P + Q$. When reparametrising an arena with players, it is not always clear if and how players' identities are preserved.
Often though, we know how we want players to change and we reparametrise along a special kind of lens that just rearranges a strategy profile to present it as if the players involved were different. We call this operation \emph{regrouping}.
Suppose we are given an arena $\A : \Arena{\Pi_{p \in P} \Omega_p}{\Pi_{p \in P} \Comega_p}{X}{S}{Y}{R}$ with players~$P$, and a function $r : P \to Q$.
Define $Q$-indexed families $r_*\Omega$ and $r_*\Comega$ by
\begin{equation*}
  (r_*\Omega)_q = \prod_{r(p) = q} \Omega_p
  \text{ and }
  (r_*\Comega)_q = \prod_{r(p) = q} \Comega_p.
\end{equation*}
Then regrouping along~$r$ is reparametrisation along the lens
\[
  w : \Lens{\Pi_{q \in Q} (r_*\Omega)_q}{\Pi_{q \in Q} (r_*\Comega)_q}{\Pi_{p \in P} \Omega_p}{\Pi_{p \in P} \Comega_p}
\]
whose ``get'' and ``put'' parts are the canonical permutations. Note that $r$ does not need to be injective or surjective for this to be well-defined.

\begin{example}[Dummy players]
  One particular use-case of regrouping is along an inclusion map $r : P \hookrightarrow P'$. For elements of $P'$ not in the image of $r$, this introduces spurious trivial factors in the strategy profiles and rewards through the empty products $\Pi_{r(p) = p'} \Omega_p$, and only changes the indexing of the strategy profiles to a larger set. This can be seen as introducing `dummy players' with exactly one strategy, who do not participate in the interaction described by the given arena.
\end{example}

\begin{example}[Identifying players]
\label{ex:identifying-players}
    Let $\A$ be an arena whose set of players is of the form $P+P$.
    This happens for example when we compose (sequentially or in parallel) two games with the same set of players.
    We can turn it into a $P$-player arena by regrouping along the codiagonal $\nabla = [\id, \id] : P + P \to P$.
    \begin{figure}[H]
      \sctikzfig[1.2]{regrouping-alt}
    \end{figure}
    By definition, the total set of strategy profiles will not change, but the two players $(\inl\; p)$ and $(\inr\; p)$ are now identified as a single player~$p$ with set of strategies $\Omega_{p}^{r_* \A} = \Omega^\A_{\inl\; p} \times \Omega^\A_{\inr\; p}$.
    This operation is fundamental to actually `put players in arena', since it breaks the locality of identities by making long-range identifications of decisions.
    The picture above represents the situation where we sequentially compose two arenas~$\A$ and~$\A'$ with the same set of players, $P = \{p_1, p_2\}$.
    The composite arena $\A \fatsemi \A'$ has four players, but we can regroup them along $\nabla : P+P \to P$ to get a two-player game again.
    This example makes it apparent that regrouping players simply consists in re-ordering the wires of the strategies and rewards.
\end{example}

\paragraph{External choice in arenas.} %
Given an $I$-indexed family of arenas $A_i$, the \emph{external choice} operator $\extch_{i \in I}$ produces a new arena where the environment decides which of the arenas is going to be played:
\begin{mathpar}
  \inferrule{
    \A : (i : I) \to \Arena{\Omega_i}{\Comega_i}{X_i}{S_i}{Y_i}{R}
  }{
    \bigextch_{i \in I} \A_i : \Arena{\Pi_{i \in I} \Omega_i}{\Sigma_{i \in I} \Comega_i}{\Sigma_{i \in I} X_i}{\Sigma_{i \in I} S_i}{\Sigma_{i \in I} Y_i}{R}
  }
\end{mathpar}
Intuitively, players in $\bigextch_{i \in I} \A_i$ must be ready to play any of the arenas~$\A_i$ (as chosen by the state received from the environment), so they have a strategy for \emph{each} of them; and, after playing, they receive a reward coming from \emph{just} the arena~$\A_i$ chosen by the environment.
If all the arenas $A_i$ have the same set of players~$P$, then we can use the canonical function $\Sigma_{i \in I} \Pi_{p \in P}\Comega_i \to \Pi_{p \in P} \Sigma_{i \in I} \Comega_i$ to reparametrise $\bigextch_{i \in I} \A_i$ into an arena with players $P$.
The arena $\bigextch_{i \in I} \A_i$ is defined as follows:
\begin{align*}
  & \play_{\bigextch_{i \in I} \A_i} : \Pi_{i \in I} \Omega_i \times \Sigma_{i \in I} X_i \to \Sigma_{i \in I} Y_i\\
  & \play_{\bigextch_{i \in I} \A_i} \; \varphi \; (i,\, x) = (i,\, \play_{\A_i}\; \varphi_i\; x)\\[2ex]
  & \coplay_{\bigextch_{i \in I} \A_i} : \Pi_{i \in I} \Omega_i \times \Sigma_{i \in I} X_i \times R \rightarrow \Sigma_{i \in I} \Comega_i \times \Sigma_{i\in I} S_i\\
  & \coplay_{\bigextch_{i \in I} \A_i} \; \varphi \; (i, x) \; r =  ((i, u) , (i, s))\ \msf{where}\ (u, s) = \coplay_{\A_i}\; \varphi_i\; x\; r
\end{align*}

\noindent The only component of the games $\G_i$ that is not allowed to depend on~$i$ is the set of utility, $R$. This is because coproducts in $\msf{Lens}$ exist only for this kind of family. Lifting this requirement requires to move to a dependently-typed framework, which we leave for future work.

\begin{example}[Stopped arena]
  Suppose $\A : \Arena{\Omega}{\Comega}{X}{S}{Y}{R}$ is an arena,
  and consider the $0$-player arena $\stop : \Arena{\One}{\One}{\One}{S}{\One}{R}$, which essentially consists of a function $R \to S$.
  We can build the arena $A\,\extch\,\stop$, where the environment can decide either to play in the arena~$A$ or to stop.
  Such an arena is usually pre-composed by an arena $\switch : \Arena{\Two}{S}{X}{S}{X+\One}{S+S}$ that consists of a new player who gets to decide, through a binary choice $\Two := \{\stop,\go\}$, whether the interaction should continue as prescribed by $\A$ or stop there.
  The play function of $\switch$ either forwards the state $x$ to the arena $A$, or decides to stop, depending on the strategy.
  The coplay function is a mere codiagonal $\nabla$.
  This situation can be represented as follows:
  \begin{figure}[H]
    \sctikzfig[1.2]{stop-game}
  \end{figure}
\end{example}

\paragraph{Internal choice in arenas.} %
Given an $I$-indexed family of arenas $A_i$, the \emph{internal choice} operator $\intch_{i \in I}$ produces a new arena where the players decide which of the arenas is going to be played:
\begin{mathpar}
  \inferrule{
    \A : (i : I) \to \Arena{\Omega_i}{\Comega_i}{X_i}{S_i}{Y_i}{R}
  }{
    \bigintch_{i \in I} \A_i : \Arena{\Sigma_{i \in I} \Omega_i}{\Sigma_{i \in I} \Comega_i}{\Pi_{i \in I} X_i}{\Sigma_{i \in I} S_i}{\Sigma_{i \in I} Y_i}{R}
  }
\end{mathpar}
Intuitively, players in $\bigintch_{i \in I} \A_i$ can choose any of the arenas $\A_i$ (for each of which the environment provides a state), so they have a strategy for \emph{just one} of them; and, after playing, they propagate the reward coming from \emph{just} the chosen arena $\A_i$.
This operation is compositional at the level of arenas, but not at the level of arenas with players: there is no canonical way to decide \emph{which} player gets to choose in which arena to play. Hence this information must be provided explicitly in the form of a reparametrisation lens if one again wants an arena with players, as can be seen e.g.\ in \cref{ex:voting} below.
The arena $\bigintch_{i \in I} \A_i$ is defined as follows:
\begin{align*}
  & \play_{\bigintch_{i \in I} \A_i} : \Sigma_{i \in I} \Omega_i \times \Pi_{i \in I} X_i \to \Sigma_{i \in I} Y_i\\
  & \play_{\bigintch_{i \in I} \A_i} \;(i,\, \omega) \; x = (i,\, \play_{\A_i}\; \omega\; x_i)\\[2ex]
  & \coplay_{\bigintch_{i \in I} \A_i} : \Sigma_{i \in I} \Omega_i \times \Pi_{i \in I} X_i \times R \rightarrow \Sigma_{i \in I} \Comega_i \times \Sigma_{i\in I} S_i\\
  & \coplay_{\bigintch_{i \in I} \A_i} \; (i,\, \omega) \; x\; r = ((i, u), (i, s))\ \msf{where}\ (u, s) = \coplay_{\A_i}\; \omega\; x_i\; r
\end{align*}

\begin{example}[Voting]
  \label{ex:voting}
  Consider an $I$-indexed family of arenas $\A_i$, each with players $P$ and rewards~$R$. One way in which players can decide which arena to play in is to vote.
  For each $p \in P$, define $\Omega'_{p} = I \times \Pi_{i \in I}\Omega_{(i,p)}$, where the first factor represents the vote for which arena to play in, and the second factor is a strategy for $p$ to play in all arenas $\A_i$. Let also $i_0 \in I$ be a distinguished option, which we default to if no majority is reached (for example, $\A_{i_0}$ may be an arena in which no game is played and all players instead receive a penalty).
  We define
  \begin{equation*}
    \vote : \Lens{\Pi_{p \in P} \Omega'_{p}}{R^{P}}{\Sigma_{i \in I} \Pi_{p \in P} \Omega_{i,p}}{\Sigma_{i \in I} R^{P}}
  \end{equation*}
  by the following ``get'' and ``put'' maps:
  \begin{eqnarray*}
    \get_{\vote}\; \omega' & = & \begin{cases}
      (i,\, \lambda p\,.\, (\pi_2\; (\omega'\; p))\; i) & \text{if $i$ is the majority vote in $\omega'$}\\
      (i_0,\, \lambda p \,.\, (\pi_2\; (\omega'\; p))\; i_0) & \text{if there is no majority vote in $\omega'$}
    \end{cases}\\
    \put_{\vote} \; \omega'\; (i, r) &=& r
  \end{eqnarray*}
  By reparametrising the arena $\bigintch_{i \in I} \A_i$ along $\vote$, we thus get an arena with players $P$ again.
\end{example}

\paragraph{Operators on games, and a recipe for building open games.}
\label{sec:recipe}

As mentioned before, we introduced operators on arenas since it is in general not possible to define them on games. The culprit lies in the fact that preferences, as represented by selection functions, are not compositional: tensoring selection functions is not the same thing as having a joint selection on the product.
Even if sequential and parallel composition lift straightforwardly from arenas (with players) to games (with players), reparametrisation and both kinds of choice do not do so in any canonical way.
The general way of building games proceeds therefore in two separate steps:
\begin{enumerate}
  \item Use the abundant compositional setting of parametrised lenses to build the arena of the game, in particular using reparametrisation (or just regrouping) operations to link different decisions to the same players.
  The parametrisation lens describes how individual player strategies and rewards are related to the overall strategies and rewards of the arena.
  \item Assign selection functions to each player, and use $\boxtimes$ to put them all together into a global selection function for the game.
  This has to be done last so that we can choose selection functions on the adequate space of strategies.
  Improving this step is the object of ongoing work by the authors.
\end{enumerate}

%%% Local Variables:
%%% mode: latex
%%% TeX-master: "../main.tex"
%%% End:

%% file: src/translation.tex
\section{The translation from extensive form games to open games}
\label{sec:ext_translation}

We now employ the operators we just defined in order to translate extensive form games (\cref{sec:ext_games}), both with perfect and imperfect information, into open games with players.

\paragraph{With perfect information}
Following the recipe at the end of \cref{sec:recipe}, our translation proceeds from the arena up.
We build the arena associated to a given $\PETree$ by following the compositional (inductive) structure of the data type: every node will be modelled by a single decision composed sequentially with the external choice of the arenas corresponding to the subsequent branches.

A single-player decision is a very simple arena $\Dec{M}{R} = (\id, \Delta) : \Arena{M}{R}{1}{R}{M}{R}$ where the player gets to choose a move in~$M$, and receives some utility in~$R$. This utility is copied and serves as a reward for the player; it is also back-propagated to the previous player as coutility.
These simple arenas are the basic blocks of our translation; each node in the tree gives rise to a single-player decision.
\begin{eqnarray*}
    && \PETtoArena : (T : \PETree) \to \Arena{\Pi_{p \in P} \strategies\; T\;p}{R^P}{\One}{R^P}{\paths_\PET\; T}{R^P} \\
    && \PETtoArena \; (\Leaf \; v) = (\id, \Delta) \\
    && \PETtoArena \; (\Node \; p \; n \; f) = \reparam{(\regroup_p \fatsemi \discard)}{A}\\
    && \qquad \qquad \msf{where}\quad \A = \Dec{[n]}{R^P} \ \fatsemi (\id, \nabla) \fatsemi \bigextch_{m \in [n]} \; (\PETtoArena\; (f\; m))
\end{eqnarray*}
Here, $\nabla : \Sigma_{m \in [n]} R^P \to R^P$ denotes the codiagonal $\nabla = [\id, \ldots, \id]$.
The lens $\discard = (\id, \pi_1)$ is the identity on strategies, and the first projection $\pi_1 : R^P \times \Sigma_{m \in [n]} R^P \to R^P$ on rewards.
The lens $\regroup_p$ essentially regroups players along the map $[\subseteq, \id] : \{p\}+P \to P$, as in \cref{ex:identifying-players}, and redistributes rewards in $R^P$ to each player.

\begin{figure}[h]
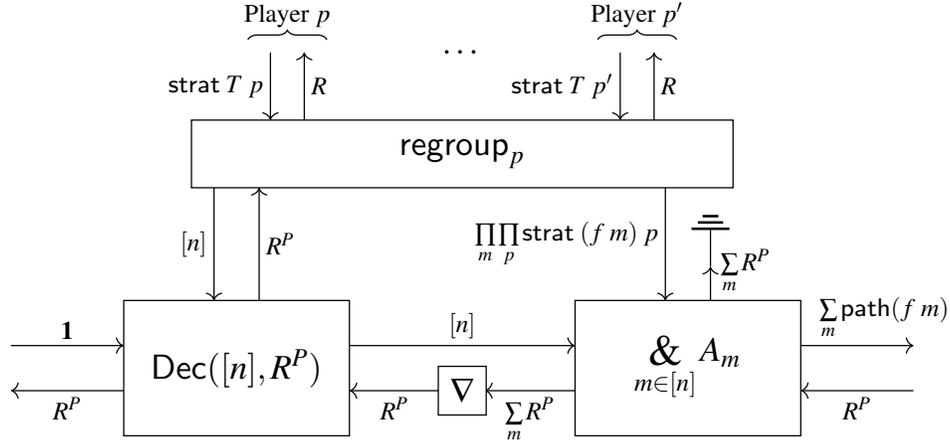

    \sctikzfig[1.2]{translation-PET}
    \caption{The inductive case of $\PETtoArena$.}
    \label{fig:pet-to-arena}
\end{figure}

Figure~\ref{fig:pet-to-arena} should help the reader typecheck the case ``$\Node\; p \; n \; f$'' of the inductive definition above, where we have
written $A_m = \PETtoArena\; (f\; m)$ for the arena obtained by recursively translating the $m$-th subtree of~$T$, of type $A_m : \Arena{\profiles (f\;m)}{R^P}{\One}{R^P}{\paths_\PET (f\;m)}{R^P}$.

%\pagebreak

\begin{theorem}
    Let $T \in \PETree$ be an extensive form game, with set of players $P$ and rewards in $R$.
    Let~$\G_T$ be the open game with players whose arena is $\PETtoArena\;T$, and whose selection function is
    \begin{equation*}
        \Boxtimes_{p \in P} \argmax^R_{\strategies\;T\;p}.
    \end{equation*}
    Denote by~$\unit$ the only state of $\G_T$. Then the set of equilibria of $\G_T$ for the context $(\unit, \payoff_\PET\; T)$ coincides with the set of Nash equilibria of $T$. \qed
\end{theorem}
% \begin{proof}
%     By construction, $T$ and $\G_T$ have the same set of strategy profiles and the same utility function.
%     \fnf{I'm not so convinced by the following sentence --- there isn't really a normal form for an arbitrary open game, right?}
%     By common reduction to normal form, the claim follows from \cref{th:normal_form_nash}.
% \end{proof}

\paragraph{With imperfect information}
To translate an $\IETree$ into an open game with players, we proceed in three steps: first we create a perfect information game, which we know how to translate to an arena with players from the previous subsection. Then we reparametrise to coordinate the decisions of a player among all nodes labelled with the same information set (similar to \cref{ex:diag}). Finally we add in selection functions as a Nash product of $\argmax$ for each player, as before.
Throughout this section, we fix a set of players $P$ and information sets $(I, \belongs, n)$.
The first step is defined by:
\begin{eqnarray*}
    && \IETtoPET : \IETree \to \PETree\\
    && \IETtoPET\; (\Leaf\; v) = \Leaf\; v\\
    && \IETtoPET\; (\Node\; i\; f) = \Node\; (\belongs\,i)\; (n\; i)\; (\lambda m \,.\, \IETtoPET\; (f\; m))
\end{eqnarray*}
Using $\PETtoArena$, we can then translate this game to an arena with players, but its set of strategies $\strategies_{\PET}\,(\IETtoPET\, T)$ are not correct, as they may allow a player to play different moves at nodes belonging to the same information set. Hence we reparametrise this arena along the lens
\[
    \clone(T) : \Lens{\profiles_\IET\,T}{R^P}{\Pi_{p \in P}\strategies_{\PET\,}\,(\IETtoPET\, T)\,p}{R^P},
\]
whose ``put'' map is the second projection, and whose ``get'' for player $p$ is defined by induction over $T$:
\begin{align*}
  & \get_{\clone(\Leaf\,u)}\; p \; \bullet = \bullet \\
  & \get_{\clone(\Node\,i\,f)}\; p \; (m, \varphi) =
    (\msf{if}\ \belongs\; i = p\ \msf{then}\ m\ \msf{else}\ \bullet,\; \lambda m' \,.\, \get_{\clone(f\,m')}\,p\,(\varphi\, m'))
\end{align*}
This clones the move $m$ to all information sets belonging to $p$.

\begin{example}
    To illustrate the algorithm just described, we translate the right tree of \cref{ex:IET}.
    As instructed above, we first translate the underlying tree without taking into account the information sets, and then reindex the set of strategies of the second player in order to copy the decision $\{\ell,r\}$ in both nodes of the information set.
    The resulting game is depicted below.
    \begin{figure}[H]
        \sctikzfig[1.2]{translation-example2-alt}
    \end{figure}
    In the picture, black dots indicate copy maps, so the only substantial part is the function labelled ``$\play$'', of type $\play : {\{L,R\} \times \{\ell,r\} \times \{\ell,r\} \to \{L\ell,Lr\} + \{R\ell,Rr\}}$.
    It takes as input the choice in $\{L,R\}$ made by Player 1 at the root of the tree, and the two choices in $\{\ell,r\}$ made by Player 2 in the two subtrees (in fact, these two choices are equal because of the copy map contained in the reparametrisation ``$\clone$'').
    The output is the corresponding path in the extensive form tree: $L\ell, Lr, R\ell$, or $Rr$.
    The payoff matrix of the game is not represented in the picture: it can be encoded as a costate $k : \{L\ell, Lr, R\ell, Rr\} \to \mathbb{R}^2$, to be plugged at the bottom right.
\end{example}

\begin{theorem}
    Let $T \in \IETree$ be an extensive form game with imperfect information, with set of players $P$ and rewards in $R$.
    Let $\G_T$ be the open game with arena $\reparam{\clone}(\PETtoArena\,(\IETtoPET\, T)$
    and selection function
    \begin{equation*}
        \Boxtimes_{p \in P} \argmax^R_{\strategies_\IET\;T\;p}.
    \end{equation*}
    Denote by~$\unit$ the only state of $\G_T$. Then the set of equilibria of $\G_T$ for the context $(\unit, \payoff_\PET\; T)$ coincides with the set of Nash equilibria of $T$. \qed
\end{theorem}
%\begin{proof}
%  By construction, $T$ and $\G_T$ have the same set of strategy profiles and the same utility function, and the play function of $\G_T$ matches the play function of $T$. By inspection, the equilibria are the same.
%\end{proof}

%%% Local Variables:
%%% mode: latex
%%% TeX-master: "../main"
%%% End:

%% file: src/conclusion.tex
\section{Concluding remarks and future work}
\label{sec:conc}

This paper is an experiment in refining the notion of open game to include an explicit treatment of the set of players.
The main technical step that distinguishes our approach from previous treatments of open games is the decoupling of the selection functions from the \emph{arena} of the game.
In particular, the action of \emph{regrouping} players, or more generally \emph{reparametrising} the sets of strategy profiles and reward vectors, is done at the level of arenas, without taking into account the selection functions.

As a sanity check, we have shown how to translate an extensive form game (with possibly imperfect information) to an open game with agency, in a way that respects the compositional structure of players' decisions.
Open games obtained by such translations are but a small subclass of all open games --- for instance, they use $\coplay$ in trivial ways.
Another way to go beyond extensive form games is to use reparametrisation: this allows, for example, to create coalitions of players sharing rewards (\cref{ex:addRewards}), or to assign decisions to players according to a vote (\cref{ex:voting}).
Lastly, in extensive form games, information sets are sometimes used to simulate simultaneity (see \cref{ex:IET}), but open games have a parallel composition operator~$\otimes$ that can model it more explicitly.

Several research directions are left for future work.
One of them is to extend this work to extensive form games built on infinite trees, possibly adapting the construction for repeated open games~\cite{GhaniKLF18}.
We would also like to extend our work to other solution concepts than Nash equilibria: these include subgame-perfect equilibria, sequential equilibria, or evolutionarily stable strategies.
Another unsettled question is to understand how selection functions interact with reparametrisations and, most importantly, what the mathematical nature of apparently uniform assignments of selection functions such as $\argmax$ is (which would allow to lift arenas to games, thereby promoting the compositionality of the former to compositionality of the latter).

Finally, we note that the formalism used here, based on parametrised lenses, is an instance of a general framework dubbed \emph{categorical cybernetics} in a companion paper in this conference~\cite{cybernetics}.

%%% Local Variables:
%%% mode: latex
%%% TeX-master: "../main"
%%% End: